\documentclass[11pt,rqno, a4paper]{article}

\usepackage{xcolor}
\pagecolor{white}

\usepackage{titlesec}
\titleformat*{\section}{\sc\centering\large} 
\titleformat*{\subsection}{\bf} 
\titleformat*{\subsubsection}{\it} 

\usepackage[small]{caption}
\usepackage{amsmath}
\usepackage{amssymb}
\usepackage{amsfonts}
\usepackage{setspace}
\usepackage{latexsym}
\usepackage{mathrsfs}
\usepackage{epsfig}
\usepackage{amsthm}
\usepackage{color}

\usepackage[margin=2cm]{caption}

\usepackage{comment}

\newcommand{\be}{\begin{equation}}
\newcommand{\ee}{\end{equation}}

\newcommand{\vs}{\vspace{0.2cm}}

\usepackage{fancyhdr}
\newtheorem{Theorem}{Theorem}

\newtheorem{Proposition}{Proposition}

\newtheorem{Estimate}{Estimate}

\newcommand{\Sa}{{\rm S}^{1}} 

\usepackage{tocloft} 

\usepackage{setspace, tocloft}

\allowdisplaybreaks

\begin{document}

\thispagestyle{empty}

\begin{center}
{\Large\bf Static vacuum $3+1$ black holes that cannot be put into stationary rotation}
\vs\vs

{\sc Javier Peraza}

{javier.perazamartiarena@concordia.ca}
\vs

{\it Mathematics and Statistics Department, Concordia University} 

{\it Montreal, Canada}

\vs

{\sc Mart\'in Reiris}

{mreiris@cmat.edu.uy}
\vs

{\it Centro de Matem\'atica, Universidad de la Rep\'ublica} 

{\it Montevideo, Uruguay}

\begin{abstract}
We prove that some of the static Myers/Korotkin-Nicolai (MKN) vacuum 3+1 static black holes cannot be put into stationary rotation. Namely, they cannot be deformed into axisymmetric stationary vacuum black holes with non-zero angular momentum. We also prove that this occurs in particular for those MKN solutions for which the distance along the axis between the two poles of the horizon is sufficiently small compared to the square root of its area. The MKN solutions, sometimes called periodic Schwarzschild, are physically regular, have no struts or singularite, but are asymptotically Kasner.

The static rigidity presented here appears to be the first in the literature of General Relativity. 
\end{abstract}

\end{center}

\section{Introduction}
The static Schwarzschild black holes can be deformed into rotating Kerr black holes. Without this basic property, stability of the Schwarzschild black holes would not have a chance. In higher dimensions, the Tangherlini black holes, which are generalizations of Schwarzschild, can be deformed into the stationary Myers-Perry black holes that can rotate along all possible planes. Here, we study quotients with one horizon of the $3+1$ static vacuum Myers/Korotkin-Nicolai solutions (see later), and show that some of them cannot be deformed into stationary rotating ones. In other words, some of these black holes solutions, that we will simply call MKN, cannot be put into stationary rotation. What is remarkable here is the fact that except for axisymmetry (that allows us to define angular momentum) and a certain metric completeness at spatial infinity, no other assumptions are made on the deformations, in particular on their asymptotics. These MKN static black holes are genuinely isolated from stationary rotating ones. They can be deformed of course into other MKN static solutions, but none of them rotate. It is important to stress, that these quotients are physical regular solutions, free of singularities outside the horizons and complete at infinity. They are regular solutions but not asymptotically flat.

This curious result, which seems to be the first of its kind in the literature, further supports the notion that the phenomenology of General Relativity in the 3+1 asymptotically flat context is indeed quite exceptional. Thinking from this perspective, it helps to highlight the contrast between the classical 3+1-asymptotically flat scenario and other scenarios, like higher dimensions and different topologies and asymptotics. 
\vs

\subsection{Static Periodic Black Holes}

The static vacuum solutions in question in this article are quotients of the periodic solutions found first by Myers in 1987 \cite{Myers:1987qx} and rediscovered by Korotkin and Nicolai in 1994 \cite{KN}. These quotient spacetimes, that will be described more in detail below, have a single spherical horizon and their spatial topology is that of an open solid torus ($\Sa\times D^{2}$, where $D^{2}$ is an open disc) minus an open 3-ball (whose boundary is the horizon). A schematic picture of the topology is given in Figure \ref{Fig1}. They are asymptotically Kasner (see (\ref{KASA})) rather than asymptotically flat and their spatial hypersurfaces ($t=$ const) are metrically complete. As we will see, both the topology and the asymptotic (which in turn can be proved to be a consequence of the topology), are crucial to obtain the mentioned static rigidity.   
\begin{figure}[h]
\centering
\includegraphics[width=6cm, height=4cm]{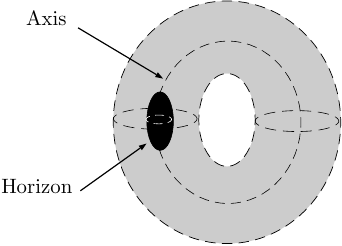}
\caption{The grey region is the spatial manifold of a MKN solution. It is topologically an open solid 3-torus minus a 3-ball (the black ball). The boundary of the black 3-ball is the horizon. The border of the solid torus, marked with a dashed line, is `infinity' and of course lies at an infinite metric distance from the horizon. The axis of the $\Sa$-symmetry is marked with the dashed line in the middle of the solid torus.}
\label{Fig1}
\end{figure}

Let's present now a brief account on the solutions found by Myers and Korotkin-Nicolai that we will quotient. Roughly speaking, they represent a configuration of infinite horizons periodically aligned on a single axis (see Figure \ref{Fig2}). These spacetimes are simply connected. We will call them `universal MKN solutions' to distinguish them from their quotients that we called MKN.
\begin{figure}[h]
\centering
\includegraphics[width=2cm, height=8cm]{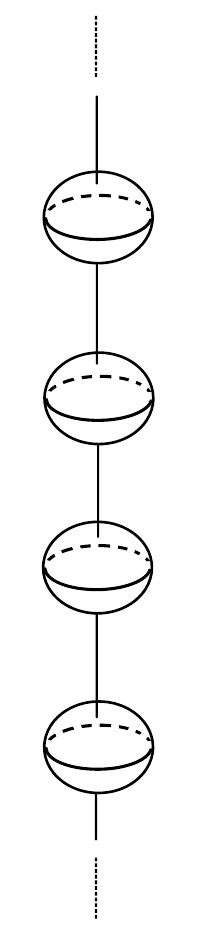}
\caption{A representation of the the infinite periodic array of horizons of the universal MKN solutions.}
\label{Fig2}
\end{figure}
In Weyl-Papapetrou coordinates the metrics of the universal MKN solutions take the form,
\be\label{MKNM}
{\mathfrak g}=-e^{\omega}dt^{2}+e^{-\omega}(e^{2k}(dz^{2}+d\rho^{2})+\rho^{2}d\varphi^{2}).
\ee
Here the axial coordinate is $\varphi$, $(z,\rho)\in \mathbb{R}_{z}\times \mathbb{R}^{+}_{\rho}$ and the exponents $\omega$ and $k$ depend only on $(z,\rho)$, (the form (\ref{MKNM}) is the one used in \cite{KN} but later we will use the form given in \cite{Wei92}, which is equivalent). The metric depends explicitly on two parameters, the `period' $L>0$ and the `horizon semi-length' $m>0$, with $L/2>m$ to avoid overlapping of horizons (that would result in non regular solutions). Note therefore that $4m/L$ can take any value between $0$ and $2$. This will be used later when we discuss a heuristic argument behind Theorem \ref{MAIN1}. 
The horizons are located on each of the closed segments of the set $\Gamma=\{(z,0): -m+jL\leq z\leq m+jL, j\in \mathbb{Z}\}$, (Weyl-Papapetrou's coordinates degenerate on the horizons). The exponents $\omega=\omega_{\rm MKN}$ and $k=k_{\rm MKN}$ are suitably singular on them. Furthermore, as $\rho\rightarrow \infty$ the metric becomes asymptotically Kasner, namely, it takes the asymptotic form,
\be\label{KASA}
{\mathfrak g}\sim -c^{2}\rho^{\alpha}dt^{2}+a^{2}\rho^{\alpha^{2}/2-\alpha}(dz^{2}+d\rho^{2})+b^{2}\rho^{2-\alpha}d\varphi^{2},
\ee
where (it can be seen) $\alpha=4m/L$ and therefore $0<\alpha<2$, and where $a, b$ and $c$ are some positive constants. The parameter $\alpha$ is known as the \emph{Kasner exponent} associated to the solution. 

The universal MKN solutions are constructed briefly as follows. For metrics of the form (\ref{MKNM}), the vacuum Einstein equations reduce to the following linear equation for $\omega$,
\be\label{LAPSEE}
\omega_{zz}+\omega_{\rho\rho}+\frac{1}{\rho}\omega_{\rho}=0,
\ee
and the following quadratures for $k$,
\be\label{CONSTRE}
k_{\rho}=\frac{\rho}{4}(\omega_{\rho}^{2}-\omega_{z}^{2}),\quad k_{z}=\frac{\rho}{2}\omega_{z}\omega_{\rho}.
\ee
Equations (\ref{LAPSEE}) and (\ref{CONSTRE}) are essentially the lapse and the constraint equations respectively. For the Schwarzschild metric of mass $m$, $\omega$ takes the form,
\be\label{omega_Sch}
\omega=\omega_{S}:=\ln\left(\frac{\sqrt{(z-m)^{2}+\rho^{2}}+\sqrt{(z+m)^{2}+\rho^{2}}-2m}{\sqrt{(z-m)^{2}+\rho^{2}}+\sqrt{(z+m)^{2}+\rho^{2}}+2m}\right),
\ee
which is singular on the segment $\{(z,0):-m\leq z\leq m\}$ where the horizon lies. To construct the universal MKN solution of period $L$ and horizon length $m$, one then considers the translations $\omega_{S}(z+jL,\rho)$, with $L/2>m$ and $j\in \mathbb{Z}$, and adds them up with a convenient counter-term to make the series convergent. The result is a periodic $\omega_{\rm MKN}$ with period  $L$. One can then prove that $k_{\rm MKN}$, found from (\ref{CONSTRE}), is also periodic. The constant of integration can then be fixed so that no struts show up in between the horizons. These solutions have no singularity outside the horizon, and are regular everywhere. We refer the reader to \cite{KN} for full details. 

The Weyl-Papapetrou coordinates are determined by the metric up to a free scaling constant. Just for reference, if $T$ and $R$ are the stationary and rotational Killing fields respectively, then $\rho=(-\mathfrak{g}(T,T)\mathfrak{g}(R,R))^{1/2}$ and $z$ is the harmonic conjugate of $\rho$. In contrast to the the asymptotically flat case, there is no natural normalization for $T$, thus both $\rho$ and $z$ get defined up to a free positive factor. Under such scaling factor, both $L$ and $m$ scale the same way so that the quotient $L/m$ is a geometric invariant. This invariant will be central in this article.
\begin{figure}[h]
\centering
\includegraphics[width=10cm, height=7cm]{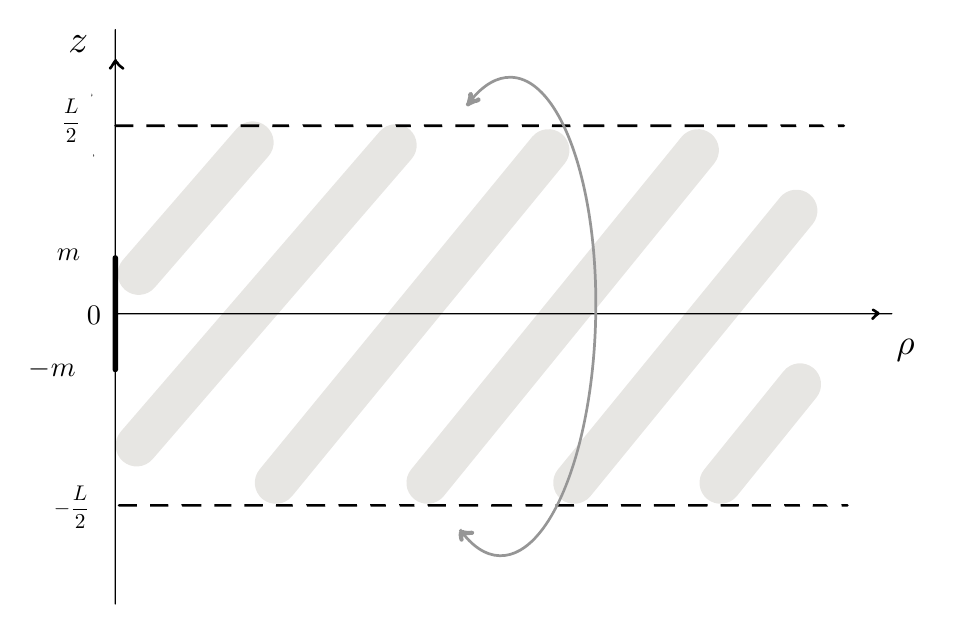}
\caption{The domains $D_{L,m}$.}
\label{Fig3}
\end{figure}

\subsection{Stationary Perturbations}

As said before, we will consider quotients of the universal MKN solutions by a translation, so that the resulting spacetime has only one horizon. The MKN black holes are $\Sa$-symmetric, but being static, have zero (Komar) angular momentum. Observe that, after the quotient, the Weyl-Papapetrou coordinate $z$ ranges between $-L/2$ and $L/2$. The points $(-L/2,\rho)$ and $(L/2,\rho)$ are identified. The quotient manifold $\Sa_{L} \times \mathbb{R}^{+}$ where $\Sa_{L}$ is a circle of length $L$, will be denoted by $S$ in the proof of Theorem \ref{MAIN1}. A representation of these domains, that we will denote as $D_{L,m}$, is given in Figure \ref{Fig3}. Observe that $S$ is a two-manifold with boundary, given by $\partial S = \Sa_L \times \{0\}$.      
\vs

We now ask whether a given MKN black hole spacetime can be put into stationary rotation. Precisely, we ask if one can find a 1-parameter family of stationary and axisymmetric black hole metrics in Weyl-Papapetrou coordinates,
\be\label{MWPC}
\mathfrak{g}_{\lambda} = -V_{\lambda}dt^{2}+2W_{\lambda}d\varphi dt+\eta_{\lambda}d\varphi^{2}+\frac{e^{2\gamma_{\lambda}}}{\eta_{\lambda}}(d\rho^{2}+dz^{2}),
\ee
over a 1-parameter family of domains $\mathcal{D}_{L_{\lambda},m_{\lambda}}$, with $L_{\lambda}$ and $m_{\lambda}$ varying continuously, and with non-zero angular momentum for all $\lambda\neq 0$. We do not require the continuity of $V_{\lambda}$, $W_{\lambda}$, $\eta_{\lambda}$ and $\gamma_{\lambda}$ but just the continuity of $L_{\lambda}/m_{\lambda}$ which is a much weaker condition. Of course $L/m$ encodes significant global geometric information. Taking this into account, for $\lambda$ small, the metrics $\mathfrak{g}_{\lambda}$ and $\mathfrak{g}_{\rm MKN}=\mathfrak{g}_{0}$ could look quite different, but $L_{\lambda}/m_{\lambda}$ must be close to $L_{\rm MKN}/m_{\rm MKN}=L_{0}/m_{0}$. In particular, the asymptotic behaviours as $\rho\rightarrow \infty$ of $V_{\lambda},W_{\lambda}$ and $\gamma_{\lambda}$ could be quite different to those of the given MNK. This is what indeed occurs if we deform the MKN solution along the MKN family itself by taking $\lambda=m-m_{0}$.  Different $\mathfrak{g}_{\rm MNK}$ have different asymptotic behavior. Of course we will always require the spatial hypersurfaces ($t=$ const) to be metrically complete at infinity for all $\lambda$, \footnote{The 3-manifold endowed with the Riemannian distance is metrically complete}. 

We address all this in the following section.

\subsection{Main results}

Coming back to the question raised, as to whether one can put the MKN black holes into stationary rotation, we prove that the answer is `no' at least when $L_{\rm MKN}/m_{\rm MKN}<4$ and provided that the 2-metric $q:=e^{2\gamma}(dz^{2}+d\rho^{2})$ is complete at infinity. The impossibility will follow from the next theorem, that in itself is a non-existence result for axisymmetric rotating black holes.

\begin{Theorem}\label{MAIN1} Assume that,  
\be
\mathfrak{g} = -Vdt^{2}+2Wd\varphi dt+\eta d\varphi^{2}+\frac{e^{2\gamma}}{\eta}(d\rho^{2}+dz^{2}),
\ee
is the metric of a periodic stationary and axisymmetric black-hole with period $L$, horizon length $m$ and non-zero angular momentum. Assume also that the 2-metric $q=e^{2\gamma}(dz^{2}+d\rho^{2})$ is metrically complete at infinity. Then, $L/m\geq 4$.
\end{Theorem}

Note that the metric completeness of $q$ is granted if $\eta$ (the norm of the axisymmetric Killing field) is bounded below away from zero at infinity and if the spatial metric $q/\eta$ is metrically complete, condition that we always assume as it is the physical metric. The MKN solutions have these properties.

To deduce from the Theorem \ref{MAIN1} that a MKN solution with $L/m<4$ cannot be put into stationary rotation, just suppose that there is a deformation family with $L_{\lambda}/m_{\lambda}$ varying continuously, with non-zero angular momentum for all $\lambda\neq 0$, and metric completeness at infinity of $e^{2\gamma}(dz^{2}+d\rho^{2})$ for all $\lambda$. Then, we would have $L_{\lambda}/m_{\lambda}<4$ for $\lambda$ sufficiently small, contradicting the theorem. Thus, MKN static black hole solutions cannot be put into stationary rotation if $L/m<4$. We state this in the following theorem.

\begin{Theorem}\label{MAIN2}
MKN black holes with $L/m<4$ cannot be put into stationary rotation, along a deformation family with $L/m$ varying continuously and with $q$ metrically complete at infinity for all $\lambda$.
\end{Theorem}

The meaning of the scale invariant and geometric invariant $L/m$ isn't particularly clear but it seems intuitively related to the distance $D$ along the axis between the two poles of the horizon. The following proposition sheds light on this point. We let $A$ be the area of the MKN black holes.
\begin{Proposition}\label{P1} For any MKN black hole we have,
\be\label{RATIOB}
\frac{D^{2}}{A}\geq \frac{1}{130}(\frac{L}{2m}-1).
\ee
\end{Proposition}
We will be proving this proposition, but let's note from this that if $D/\sqrt{A}<1/12$ then $L/m<4$. Therefore, if the distance between the two poles of the horizon is too small compared to the square root of its area, then the MKN black hole cannot be put into stationary rotation. We state this in the following theorem.

\begin{Theorem}\label{MAIN3} 
MKN black holes with $D/\sqrt{A}<1/12$ cannot be put into stationary rotation.
\end{Theorem}

This inequality is by no means sharp. The proof of \eqref{RATIOB} in Proposition \ref{P1} can be refined to obtain,
\be 
D \geq \frac{\pi }{6 } \Gamma_{\min}  \left( 1 - \frac{2m}{L} \right) \sqrt{\frac{L}{2m}} \sqrt{A},
\ee
where $\Gamma_{\min}$ is the minimum of the gamma function in $\mathbb{R}^+$ \footnote{As a reference, $\frac{\pi \Gamma_{\min}}{6 \sqrt{2}} \approx 0.3279...$.}, and the bound in Theorem \ref{MAIN3} can be moved upwards.

\subsection{Physical interpretation of results}

There is an interesting heuristic argument for Theorem \ref{MAIN1} that we would like to elaborate in what follows. Assume that we have a stationary solution as in the hypothesis of Theorem \ref{MAIN1}. Take now its universal cover. As with the universal MKN solutions, that would be a configuration of infinite horizons periodically aligned on an axis but this time all the horizons rotate with the same angular momentum. Then, on physical grounds, it is expected the metric far away from the string of horizons to approach the gravitational field of a material, infinite-length, rotating cylindrical rod. Infinite rotating cylindrical rods were studied for the first time by Van Stockum in \cite{vanStockum:1937zz}, and outside the rod the metric is known to have the following closed form (Lewis, \cite{Lewis}),
\begin{gather}
\label{ASF1} V =  \frac{2a}{|w|} e^{-b} \rho^{1-a},\quad W =   \text{s}(w)e^{-b} \rho^{1-a},\\
\label{ASF2} \eta = \frac{|w|}{2a} (e^{b}\rho^{1+a}-e^{-b}\rho^{1-a}),\quad \frac{e^{2\gamma}}{\eta} = c\rho^{(a^2-1)/2},
\end{gather}
where $a$, $b$ and $c$ are constants, with $0<a<1$ and $b,c \in \mathbb{R}$, $w$ is related to the angular momentum per unit of length via $w=\frac{8J}{L}$ and the function $s(w)$ is the sign function. The lower bound on $a$ originates when solving stationary vacuum Einstein equations with cylindrical symmetry, since as $a \rightarrow 0$ the solutions degenerate to another family of possible models for $a=0$. A prior there is no upper bound for $a$, but any $a \geq 1$ cannot model the asymptotic of coaxial arrays of black holes like the ones we are considering. Indeed, $a\geq 1$ implies that the lapse function for the hypersurface $\{t = 0\}$ tends to zero at infinity, which can be ruled out by the maximum principle on the lapse equation taking into account that the lapse is zero on the horizons. We refer the reader to \cite{Peraza_2023} for a straightforward derivation and discussion of Lewis' models.

Now, for the solution we are considering (not its universal cover), we consider the Komar mass relative to the Killing field $\partial_{t}$. We evaluate it over the horizon and at `infinity' assuming the asymptotic form of the metric (\ref{ASF1})-(\ref{ASF2}) (take the limit calculated over the sequence of divergent 2-tori with constant $\rho$, as $\rho\rightarrow \infty$). Then, after a simple calculation, we obtain,
\be
(1-a)\frac{L}{4} = \frac{\kappa A}{4\pi} + 2\Omega J,
\ee
where $A$ is the area of the horizon, $\kappa$ the surface gravity, $\Omega$ the angular velocity and $J$ the Komar angular momentum. This is a Smarr-type of identity \footnote{For a detailed discussion on Smarr identity and Komar integrals, we refer the reader to Chapters 7 and 12 of Wald's textbook in General Relativity, \cite{wald84}}. The term on the left is the Komar mass at infinity and the term on the right is the usual expression in the standard Smarr identity for the Komar mass of the horizon. The angular velocity has always the same sign as the angular momentum, so the product $\Omega J$ is always positive (this term can be seen to be an explicitly positive integral). Furthermore it is easy to see that $\kappa A/4\pi = m$ (see for instance \cite{Peraza_2023}). Putting all together we deduce, 
\be\label{SI}
\frac{L}{4}>(1-a)\frac{L}{4}>m,
\ee
and therefore $L/m>4$. The only assumption in this argument is the asymptotic form of the metric coefficients in (\ref{ASF1})-(\ref{ASF2}). Proving that the asymptotic behaviour of the metric is indeed Lewis' can be an avenue to prove Theorem \ref{MAIN1}, and in fact this can be done after some considerable effort. However there is a shorter proof that does not necessitate of such a detailed analysis of the asymptotics. What our proof essentially does is to extract sufficient asymptotic information on the metric so that to make the argument above work. One final word about the argument. One may wonder if one would obtain also $L/m>4$ were one to calculate the Smarr identity for the MKN solutions. After all, the only difference between a solution as in the Theorem \ref{MAIN1} and a MKN solution is that one rotates and the other does not. However, when calculating the Smarr identity for the MKN solution, one must use the Kasner asymptotics (\ref{KASA}) rather than (\ref{ASF1})-(\ref{ASF2}), giving,
\be       
\alpha \frac{L}{4}=m.
\ee
But, we know that $\alpha$ can take any value between $0$ and $2$, so, when repeating the argument as in (\ref{SI}) we get, at best,
\be
 L/2>\alpha L/4=m.
 \ee
Therefore $L/m>2$ rather than $L/m>4$. We then see that although the asymptotic is Kasner when rotation is or isn't present, in the first case the Kasner exponent $1-a$ takes values in $(0,1)$, whereas in the second case the Kasner exponent $\alpha$ can take values in $(0,2)$. This peculiar discontinuity in the asymptotics is the fundamental fact behind Theorem \ref{MAIN1}.   
\vs

\subsection{Outlook, open problems and further work.}

There are many interesting avenues for future research. Below we mention three of them.

First, one may wonder whether MKN solutions with $L/m<4$ could be deformed into non-axisymmetric and non-static stationary solutions. While that could be a possibility, we believe that: (i) a stationary deformation should always be axisymmetric in the same way one expects that stationary asymptotically flat solutions with one horizon should be axisymmetric and therefore Kerr, and (ii) once having axisymmetry one should expect to have global Weyl-Papapetrou coordinates following standard arguments by Carter (see \cite{AST_2008__321__195_0} and references therein). We note that it is very uncertain whether the main result of this article could be generalized to higher dimensions, as in higher dimensions there could be many possible ``deformation directions" where on could put the static solution to rotate.\vs

Second, it is not known whether MKN solutions with $L/m$ sufficiently large could be put into stationary rotation. The numerical work of \cite{Peraza_2023} supports that possibility. We leave this question as an interesting open problem. 

Finally, a third interesting avenue of research would be to investigate axisymmetric rotating perturbations of MKN black holes with $L/m<4$ as an initial data evolution problem. Such perturbations, (that can actually be constructed), may not evolve into stationary solutions. A peculiar instability must be present but its nature is at the moment unclear.

\section{Proof of the main results}

\subsection{Proof of Theorem \ref{MAIN1}}

The coefficient $V$ in (\ref{MWPC}) relates to $W$ and $\eta$ by $V=(\rho^{2}-W^{2})/\eta$. On the other hand $W$ relates to $\eta$ and the so called twist potential $\omega$ by $W=\eta\Omega$ where $\Omega$ is the angular velocity function that is found from $\eta$ and $\omega$ by $d\Omega = - \rho (*d\omega)/\eta$ and where $*$ is the Hodge-star with respect to the flat metric $d\rho^{2}+dz^{2}$ \footnote{As a reference, $*d\rho - dz$ and $*dz  = d\rho$.}. Conversely, these equations can be used to define $\omega$ in terms of $W$ and $\eta$. Let $q$ the 2-metric $q = e^{2\gamma}(d\rho^{2}+dz^{2})$. Then, the data $(q;\eta, \omega)$ satisfies the closed elliptic system:
\begin{align}
\label{REE1} \Delta_{f} \eta & = \frac{|\nabla \eta|^{2}-|\nabla \omega|^{2}}{\eta},\\
\label{REE2} \Delta_{f} \omega & = 2\langle \nabla \omega, \frac{\nabla\eta}{\eta}\rangle,\\
\label{REE3} Ric - \nabla\nabla f & =  \frac{1}{2}\frac{\nabla \eta \nabla \eta + \nabla \omega \nabla \omega}{\eta^{2}} +\nabla f\nabla f,
\end{align}
where $f=\ln \rho$ and $\Delta_{f} \phi := \Delta \phi + \langle \nabla f, \nabla \phi\rangle$. All norms and inner products are with respect to $q$. These are essentially the reduced Einstein equations. In this article we will scarcely use them directly. Rather, we will mostly make use of some geometric information that has been obtained from them in \cite{Reiris1}. We discuss all that below.

The combination $Ric-\nabla\nabla f$ is the so called Bakry-\'Emery Ricci tensor, that it is non-negative by equation (\ref{REE3}). The first two equations are the $f$-harmonic map equations for the map $p\in S \rightarrow (\omega(p), \eta(p))\in \mathbb{H}^{2}$ (see \cite{Wei92}). Observe that the Gaussian curvature $\kappa$ of $q$,
\be
\kappa=\frac{|\nabla \omega|^{2}+|\nabla \eta|^{2}}{\eta^{2}},
\ee
is the energy of the harmonic map and is explicitly non-negative. The equations (\ref{REE1})-(\ref{REE3}) form a very structured system from which relevant geometric information can be deduced. For instance, after some manipulation (see \cite{Reiris1}) one can obtain $\Delta_{f} |\nabla \rho|^{2}/\rho^{2} \geq 2(|\nabla \rho|^{2}/\rho^{2})^{2}$ and $\Delta_{f} \kappa \geq 4\kappa^{2}$. In other words, either, $\psi = \kappa$ or $\psi = |\nabla \rho|^{2}/\rho^{2}$ satisfy an equation of the form $\Delta_{f}\psi\geq \psi^{2}$ for some constant $c>0$. Then, the straightforward Lemma 3.2 in \cite{Reiris:2015zaa} shows that whenever $Ric-\nabla\nabla f \geq 0$ and $\Delta_{f} \psi \geq c\psi^{2}$ then $\psi(p)\leq c'/{\rm dist}(p,\partial S)$ for some $c'>0$. Thus, the following estimates hold for $\eta,\omega$ and $\rho$.


%
\begin{Estimate}[\cite{Reiris1}, Proposition 9] \label{EST1} There is a constant $c$ independent on the data $(S;q,\eta,\omega)$ such that,
\be\label{GE}
\frac{|\nabla \omega|^{2}}{\eta^{2}}\leq \frac{c}{d^{2}},\quad \frac{|\nabla \eta|^{2}}{\eta^{2}}\leq \frac{c}{d^{2}},\quad
\frac{|\nabla \rho|^{2}}{\rho^{2}}\leq \frac{c}{d^{2}},
\ee
where the norms are with respect to $q$, and $d(p)={\rm dist}_{q}(p,\partial S)$. 
\end{Estimate}    
For the proof of Theorem \ref{MAIN1}, we are going to need a second estimate proved in \cite{Reiris1}, showing that the reverse inequality for the gradient estimate of $\rho$ is possible at least over a diverging sequence of points.
\begin{Estimate}[\cite{Reiris1}, Proposition 11] \label{EST2} There is a constant $c'$ (that may depend on the data) and a divergent sequence of points $p_{i}$, such that,
\be
\frac{|\nabla \rho|^{2}}{\rho^{2}}(p_{i})\geq \frac{c'}{d^{2}(p_{i})}.
\ee
\end{Estimate}    

There a number of geometric consequences of these estimates, proved in detail in \cite{Reiris1}, that we comment about in what follows. Estimate 1 shows that the Gaussian curvature $\kappa$ of $q$, has at most quadratic decay, namely, 
\be\label{CD}
\kappa\leq \frac{c}{d^{2}}.
\ee 
\noindent It was proved at the end of Section 3 in \cite{Reiris1} that once this curvature bound is available, there is sufficient control on the local geometry so that one can make a bootstrapping of elliptic estimates on (\ref{REE1})-(\ref{REE2}) to obtain higher order estimates for $\eta$, $\omega$ and $\rho$. The estimates are as follows,
\be\label{HGE}
\frac{|\nabla^{j} \omega|^{2}}{\eta^{2}}\leq \frac{c_{j}}{d^{2j}},\quad \frac{|\nabla^{j} \eta|^{2}}{\eta^{2}}\leq \frac{c_{j}}{d^{2j}},\quad
\frac{|\nabla^{j} \rho|^{2}}{\rho^{2}}\leq \frac{c_{j}}{d^{2j}},\quad \forall j\geq 1.
\ee
\noindent We won't need them except for $|\nabla\nabla \rho|^{2}/\rho^{2}\leq c/d^{4}$ that will be recalled a few lines below.

On the other hand, the decay (\ref{CD}) and the non-negativity of $\kappa$ implies, via the standard Bishop-Gromov volume comparison, that $(S;q)$ has either quadratic or sub-quadratic area growth, and that if the growth is quadratic then the manifold is asymptotically conical. It was proved in Proposition 12 in \cite{Reiris1}, that the presence of the positive harmonic function $\rho$ rules out a conical asymptotics, so $(S;q)$ has in fact sub-quadratic area growth.

Let now $p_{i}$ be any divergent sequence of points. Let $\rho_{i}=\rho(p_{i})$, let $\mathcal{C}_{i}$ be the embedded circle $\mathcal{C}_{i}=\{\rho=\rho_{i}\}$, $\ell_{i}={\rm length}_{q}(\mathcal{C}_{i})$ its $q$-length and $d_{i}={\rm dist}_{q}(p_{i},\partial S)$. Then, the sub-quadratic volume growth implies,
\be\label{ELLSD}
\lim_{i\rightarrow \infty} \frac{ \ell_{i}}{d_{i}} = 0.
\ee
This and the second estimate of (\ref{GE}) in turn imply,
\be\label{MEME}
\lim_{\rho_{i}\rightarrow \infty} \left(\frac{\max\{\eta(p):p\in \mathcal{C}_{i}\}}{\min\{\eta(p):p\in \mathcal{C}_{i}\}}\right) = 1.
\ee
To see this, note that if we let $s$ be the arc-length on $\mathcal{C}_{i}$, (starting from some point), then (\ref{GE}) implies $|(\ln \eta(s))'|\leq c/d(s)$. Integrating between two points $p$ and $p'$ in $\mathcal{C}_{i}$ we deduce $\ln \eta(p)/\eta(p') \leq c\int_{s}^{s'} 1/d(s)  ds \leq c\ell_{i}/\underline{d}(s)$ where $\underline{d}(s)=\min\{d(p):p\in \mathcal{C}_{i}\}$. By the triangle inequality, for any point $p \in \mathcal{C}_{i}$, we get $d_i \leq \underline{d}(s) + \ell_i $; then we have $\underline{d}(s) \geq d_{i}-\ell_{i}$. Hence, as the points $p$ and $p'$ are arbitrary, we get $1\leq \max\{\eta(p):p\in \mathcal{C}_{i}\}/\min\{\eta(p):p\in \mathcal{C}_{i}\}\leq e^{c\ell_{i}/(d_{i}-\ell_{i})}$, where the r.h.s tends to 1 by (\ref{ELLSD}).     

By definition of the sets $\mathcal{C}_{i}$, we trivially also have
\be
\lim_{\rho_{i}\rightarrow \infty} \left(\frac{\max\{\rho(p):p\in \mathcal{C}_{i}\}}{\min\{\rho(p):p\in \mathcal{C}_{i}\}}\right) = 1.
\ee

Let us specialize now to the sequence $p_{i}$ of Estimate \ref{EST2}. It was proved in \cite{Reiris1} that there exist $0<c_{1}<c_{2}$ such that 
\be\label{DB}
0< c_{1}\leq \frac{\ell_{i}\rho_{i}}{d_{i}}\leq c_{2},\quad \forall i.
\ee
(see inside the proof of Theorem 1). Intuitively, these inequalities can be seen as a confirmation of the usual heuristic argument of taking Weyl-Papapetrou coordinates as a chart over the whole manifold. Indeed, \eqref{DB} states that the lengths of the scaled transverse circles (in asymptotically flat space usually $\approx O(1)$) divided by the distance to the boundary (in asymptotically flat space usually $\approx O(\rho)$) scales as $\frac{1}{\rho}$, and therefore $\rho$ can be taken as a coordinate in the same spirit as in the asymptotically flat case.

For the sake of completeness, the argument to prove this bound is as follows. Integrate $\Delta \rho=0$ between $\mathcal{C}_{0}$ and $\mathcal{C}_{i}$ obtaining,
\be
\int_{\mathcal{C}_{i}} \nabla_{n}\rho\, d\ell = c,
\ee
where $c$ is independent on $i$, where $n$ is the outwards normal to $\mathcal{C}_{i}$. Because $n=\nabla \rho/|\nabla \rho|$ we deduce, $\int_{\mathcal{C}_{i}} |\nabla \rho|d\ell = c$. The, because $\rho$ is equal to $\rho_{i}$ on $\mathcal{C}_{i}$ we can write $\rho_{i}\int_{\mathcal{C}_{i}} |\nabla \rho|/\rho d\ell = c$. Using the mean value theorem we get,
\be
\rho_{i}\ell_{i}\frac{|\nabla \rho|}{\rho}\bigg|_{\bar{p}_{i}} = c,
\ee
where $\bar{p}_{i}$ is some point in $\mathcal{C}_{i}$. Multiply and divide now by $d_{i}$ to get,
\be\label{THIS}
\frac{\rho_{i}\ell_{i}}{d_{i}}\left(d_{i} \frac{|\nabla \rho|}{\rho}\bigg|_{\bar{p}_{i}}\right) = c,
\ee
Now, because $p_{i}$ is the sequence of Estimate 2 we have $d_{i}|\nabla\rho|/\rho|_{p_{i}}\geq c'$. Using the bound $|\nabla\nabla \rho|/\rho\leq c_{3}/d^{2}$, we will deduce next that $d_{i}|\nabla\rho|/\rho|_{\bar{p}_{i}}\geq c'/2$ for $i$ sufficiently large. To see this, note that,
\begin{equation}
\left|\nabla_{s}\frac{|\nabla\rho|^{2}}{\rho^{2}}\right| = 2 \left|\frac{\langle \nabla_{s}\nabla \rho,\nabla \rho\rangle}{\rho^{2}}\right|\leq \frac{2c_{3}}{d^{2}}\frac{|\nabla \rho|}{\rho} \leq \frac{2c_{3}c}{d^{3}}.
\end{equation}
where, to be precise, $\nabla_{s}=\nabla_{t}$ for the $t$ the tangent unit field to $\mathcal{C}_{i}\leq \frac{c''}{\underline{d}^{3}_{i}}$. Integrating between $s(\bar{p}_{i})$ and $s(p_{i})$ we get,
\be
\frac{|\nabla \rho|^{2}}{\rho^{2}}(p_{i}) - \frac{|\nabla \rho|^{2}}{\rho^{2}}(\bar{p}_{i}) \leq \left| \int_{s(\bar{p}_{i})}^{s(p_{i}} \frac{c''}{\underline{d}_{i}^{3}} d\ell\right| \leq \frac{c''\ell_{i}}{\underline{d}_{i}^{3}}\rightarrow 0.
\ee 
where the convergence to $0$ of the last term follows from \ref{ELLSD}. Thus, for $i$ sufficiently large we obtain,
\be
\frac{|\nabla \rho|^{2}}{\rho^{2}}(p_{i}) \leq 2 \frac{|\nabla \rho|^{2}}{\rho^{2}}(\bar{p}_{i}) 
\ee
as wished.


The second is direct from (\ref{THIS}) and the bound (\ref{GE}), using a uniform bound by integrating on each circle $\mathcal{C}_{i}$ the bound $d_{i}|\nabla\rho|/\rho|_{p_{i}}\leq c'$ and taking the limit $i \rightarrow 0$.
\vs

We are now ready to prove the Theorem \ref{MAIN1}. From the discussion above we will only need (\ref{GE}), (\ref{MEME}) and (\ref{DB}).

\begin{proof}[Proof of Theorem \ref{MAIN1}] We proceed by contradiction, so we assume that there is a solution with $L/m<4$. For that solution define,
\be
x:=\ln \left(\frac{\eta}{\rho}\right).
\ee
We will use this new function of $(\rho,z)$ in what follows instead of $\eta$, and later we will also use,
\be
\underline{x}(\rho):=\min_{z}x(\rho,z).
\ee
Observe that using the function $x$, $V =  \left(  \rho - \frac{1}{\rho} W^2 \right) e^{-x}$. This particular choice for $x$ will enable us to obtain a useful monotonic quantity in \eqref{int_after_bdy_cond} below.

Using (\ref{REE1}), the equation for $x$ becomes,
\be\label{xEQ}
\Delta_{\ln \rho}x=-\frac{|\nabla \omega|^{2}e^{-2x}}{\rho^{2}}.
\ee 
Integrating this equation on a bounded region between $\mathcal{C}_{\rho_{0}}$ and $\mathcal{C}_{\rho}$, for some $\rho_0 > 0$, and using the divergence theorem gives,
\be
\int_{\mathcal{C}_{\rho_{0}}}\rho_0 \nabla_{n} x\, d\ell \leq \int_{\mathcal{C}_{\rho}}\rho\nabla_{n} x\, d\ell,
\ee
where $n$ is the outwards normal to the $\mathcal{C}_{\rho}$'s. Observing that $\nabla_{n}x\, d\ell = \partial_{\rho}x\, dz$ we write this inequality as,
\be
\rho_{0}\int_{-L/2}^{L/2} \partial_{\rho}x\, dz \leq \int_{-L/2}^{L/2}\rho\partial_{\rho}x\, dz = \frac{\partial}{\partial \ln \rho}\left(\int_{-L/2}^{L/2} x(\rho,z)\, dz\right).
\ee
Now we use that the solution is a black hole. This imposes particular boundary conditions on $x(\rho,z)$ as $\rho\rightarrow 0$ (see Section 5 of \cite{Wei92} for a detailed discussion, and section 2.4, equation (24) in \cite{Peraza_2023} for the particular boundary conditions on $\sigma \equiv x - \ln \rho$). Taking that into account we deduce,
\be \label{int_after_bdy_cond}
0>L-4m \geq \frac{\partial}{\partial \ln \rho}\left(\int_{-L/2}^{L/2} x(\rho,z)\, dz\right).
\ee
Integrating this between a fixed $\ln \rho_{1}$ and $\ln \rho$, ($\rho_{1}<\rho$), we arrive at,
\be
\underline{x}(\rho)\leq \frac{1}{L}\int_{-L/2}^{L/2}x(\rho,z)\, dz \leq (\frac{4m}{L}-1)\ln\left(\frac{\rho}{\rho_{1}}\right) + \frac{1}{L}\int_{-L/2}^{L/2}x(\rho_1,z) \, dz.
\ee
Therefore,
\be
\underline{x}(\rho)\rightarrow -\infty,\quad {\rm as}\quad \rho\rightarrow \infty.
\ee

What we will prove in what follows is that $\underline{x}(\rho)$ is indeed bounded below, which gives us the desired contradiction. 

To achieve this we begin showing that $\underline{x}(\rho)$ is either monotonically increasing or it is monotonically decreasing after some $\rho_{0}$. To prove this we show first a basic claim: for no $\rho_{1}<\rho_{2} < \rho_{3}$ we can have $\underline{x}(\rho_{1})>\underline{x}(\rho_{2})$ and $\underline{x}(\rho_{2})<\underline{x}(\rho_{3})$. Assume that such $\rho_{1},\rho_{2}$ and $\rho_{3}$ exist. Then $x$ must have a local minimum in the open region between $\mathcal{C}_{\rho_{1}}$ and $\mathcal{C}_{\rho_{3}}$. However $x$ is not constant and $\Delta_{\rho} x\leq 0$, and we contradict the maximum principle. Assume now that $\underline{x}$ is not monotonically increasing. Then there are $\rho_{1}<\rho_{2}$ such that $\underline{x}(\rho_{1})>\underline{x}(\rho_{2})$. Now, if $\rho_{3}>\rho_{2}$ then, because of the claim just proved, we must have $\underline{x}(\rho_{2})>\underline{x}(\rho_{3})$ and, again, if $\rho_{4}>\rho_{3}$ it must be $\underline{x}(\rho_{3})>\underline{x}(\rho_{4})$. Thus, for any $\rho_{4}>\rho_{3}>\rho_{2}=:\rho_{0}$ we must have $\underline{x}(\rho_{3})>\underline{x}(\rho_{4})$, proving that it is monotonically decreasing after $\rho_{0}$. 

We move now to prove that $\underline{x}$ is bounded below. By what was proved above, if $\underline{x}$ is not bounded below then it must be monotonically decreasing after some $\rho_{0}$ so it must be $\underline{x}(\rho_{i})\rightarrow -\infty$ for any sequence $\rho_{i}\rightarrow \infty$, ($\rho_{i}\geq \rho_{0}$). Given a sequence $\rho_{i}\rightarrow \infty$ let $p'_{i}\in \mathcal{C}_{\rho_{i}}$ be such that $x(p'_{i})=\underline{x}(\rho_{i})$. Then, for this sequence, it is,
\be\label{CRU}
\frac{\eta(p'_{i})}{\rho_{i}}=e^{x(p'_{i})}=e^{\underline{x}(\rho_{i})}\rightarrow 0.
\ee
We will use this information later.

Now, the Komar angular momentum can be calculated on any $\mathcal{C}_{\rho}$ as,
\be
J=\frac{1}{8}\int_{\mathcal{C}_{\rho}}\nabla_{s}\omega\, d\ell.
\ee
where $s$, in $\nabla_{s}\omega$ is a unit tangent vector to $\mathcal{C}_{\rho}$ in the direction of increasing $z$. Therefore, for any sequence $\rho_{i}\rightarrow \infty$ we can use the mean value theorem to have,
\be\label{DTO}
0\neq 8J=\ell_{i}\nabla_{s}\omega(p''_{i}),
\ee
for some $p''_{i}\in \mathcal{C}_{\rho_{i}}$ and where, recall, $\ell_{i}={\rm length}_{q}(\mathcal{C}_{\rho_{i}})$.
 
We are going to use the information we have collected so far to reach the desired contradiction. We specialize now to the divergent sequence $\rho_{i}=\rho(p_{i})$, where $p_{i}$ is the sequence of Estimate \ref{EST2}. We write (\ref{DTO}) as,
\be
0\neq 8J=\left(\frac{\ell_{i}\rho_{i}}{d''_{i}}\right)\, \left(\frac{d''_{i}\nabla_{s}\omega(p''_{i})}{\eta(p''_{i})}\right)\, \left(\frac{\eta(p'_{i})}{\rho_{i}}\right)\, \left(\frac{\eta(p''_{i})}{\eta(p'_{i})}\right),
\ee
where $d''_{i}=d(p''_{i})$ and where as earlier $p'_{i}\in \mathcal{C}_{\rho_{i}}$ is such that $\underline{x}(\rho_{i})=x(p'_{i})$. From (\ref{DB}), (\ref{MEME}) and (\ref{GE}),  we know that $(\ell_{i}\rho_{i})/d_{i}$, $\eta(p''_{i})/\eta(p'_{i})$ and $d_{i}\nabla_{s}\omega(p''_{i})/\eta(p''_{i})$ are bounded. As $\eta(p'_{i})/\rho_{i}$ tends to zero by (\ref{CRU}), we conclude that the whole expression must tend to zero, reaching thus a contradiction.
\end{proof}

\subsection{Proof of Theorem \ref{MAIN3}}

We move now to prove the Proposition \ref{P1}. Theorem \ref{MAIN3} follows as a corollary.

\begin{proof}[Proof of Proposition \ref{P1}] 
We first express $A$ in terms of $L$ and $m$. To do that it is more convenient to work with the lapse function $N =e^{\omega/2}$. We first recall that $|\nabla N|$ is constant along the horizon (equal to the surface gravity) and therefore, integrating the lapse equation $\Delta N = 0$ over the hypersurface relying on the Kasner asymptotics (\ref{KASA}) at infinity, we get, 
\be\label{AE}
A|\nabla N|_{H}=4\pi m.
\ee
This is the usual $M = \frac{\kappa A }{4\pi}$ identity for a static black hole, where $\kappa$ is the surface gravity, but written in the context of a periodic solution.

Using the explicit form (\ref{MKNM}), we can calculate $|\nabla N|$ at the pole $(z,\rho)=(m,0)$, as, 
\be
|\nabla N|_{H} = |\nabla N|(m,0) = \frac{1}{2}\frac{{\rm d}e^{\omega}}{{\rm d}z}\bigg|_{(m,0)}.
\ee
To compute it, we use the following quite useful expression for $e^{\omega}$ over the axis $\rho=0, z\in (m,L/2)$ found by Korotkin and Nicolai in \cite{KN},
\be\label{wA}
e^{\omega}=e^{4\gamma m/L}\frac{\Gamma((z+m)/L)\Gamma(1-(z-m)/L)}{\Gamma((z-m)/L)\Gamma(1-(z+m)/L)},
\ee
where here $\Gamma(x)$ is the Gamma function, and $\gamma$ is the Euler-Mascheroni constant. This surprising close formula arise from the properties of infinite product expansions for $\Gamma$ functions when constructing the universal cover of the MKN solutions as explained below equation \eqref{omega_Sch} \footnote{In particular, Weierstrass' identity involves both $\Gamma$ functions and the Euler-Mascheroni constant,
\begin{equation*}
\Gamma(z) = \frac{e^{-\gamma z}}{z} \prod_{n=1}^{\infty} \frac{n}{n + z} e^{z/n}. 
\end{equation*}
}.

We note that this expression is 0 at $z=m$ because $\Gamma((z-m)/L)$ diverges to $+\infty$ when $z\rightarrow m$. To make this more clear we transform this term using $\Gamma(1+x)=x\Gamma(x)$, valid for any $x>0$, to obtain $\Gamma((z-m)/L)=((z-m)/L)\Gamma(1+(z-m)/L)$. Replacing this expression in (\ref{wA}) and differentiating at $z=m$, we obtain,
\be
|\nabla N|_{H}=\frac{1}{2L}\frac{\Gamma(2m/L)}{\Gamma(1-2m/L)}e^{4\gamma m/L}.
\ee
Plugging this expression in (\ref{AE}) we deduce,
\be\label{AFORMULA}
A=8\pi mL \frac{\Gamma(1-2m/L)}{\Gamma(2m/L)}e^{-4\gamma m/L}.
\ee   

We will obtain now a lower bound for $D$ that we will combine with (\ref{AFORMULA}) to deduce (\ref{RATIOB}). The integral for $D$ is,
\be
D=\int_{m}^{L/2}e^{-\omega/2}dz,
\ee
where $\omega$ is evaluated on $\rho=0$, $m\leq z\leq L/2$. This integral can be bounded from below using several properties of the $\Gamma$ function that allow us to compute bounds for expressions arbitrarily close to the pole at $z=m$. First, rewriting the integrand via the use of $\Gamma(1+x) = x\Gamma(x)$ two times, we get,
\be
e^{-\omega/2}=e^{-2\gamma m/L}\bigg[\frac{\Gamma(1+\frac{z-m}{L})\Gamma(2-\frac{z+m}{L})}{\Gamma(\frac{z+m}{L})\Gamma(1-\frac{z-m}{L})}\bigg]^{1/2}\frac{1}{\sqrt{\frac{(z-m)}{L}(1-\frac{z+m}{L})}}.
\ee
To bound the middle term in square brackets we note that: (i) $1\geq (z+m)/L\geq 2m/L$, hence $\Gamma((z+m)/L)\leq \Gamma(2m/L)$, (ii) $1/2\leq 1-(z-m)/L\leq 1$, hence $\Gamma(1-(z-m)/L)\leq \Gamma(1/2)$, (iii) $\Gamma(x)\geq 3/4$, $\forall x>0$, therefore $\Gamma(1+\frac{z-m}{L})\Gamma(2-\frac{z+m}{L})\geq (3/4)^{2}$. Putting it all together, we get,
\be
D\geq e^{-2\gamma m/L} \frac{3^{2}}{4^{2}} \frac{1}{\Gamma^{1/2}(2m/L)\Gamma^{1/2}(1/2)}\int_{m}^{L/2}\frac{1}{\sqrt{\frac{(x-m)}{L}(1-\frac{x+m}{L})}} dx.
\ee
The integral is easily computed to be $\pi L/4$, and $\Gamma(1/2)\leq 2$ so we finally get the bound,
\be
D^{2}\geq e^{-4\gamma m/L} \frac{3^{4}\pi^{2}}{4^{6}}\frac{L^{2}}{2}\frac{1}{\Gamma(2m/L)}.
\ee
Combining it with (\ref{AFORMULA}) and using that $\Gamma(1-2m/L)=(\Gamma(2-2m/L))/(1-2m/L)\leq 1/(1-2m/L)$, we finally get,
\be
\frac{D^{2}}{A}\geq \frac{\pi 3^{4}}{4^{7}2}(\frac{L}{2m}-1)\geq \frac{1}{130}(\frac{L}{2m}-1).
\ee
\end{proof}

Take any MKN black hole solution with $D/\sqrt{A} < 1/12$, and assume that it can be deformed into a non-zero angular momentum stationary solution. Then, $L > 4m$ and therefore $\left( \frac{L}{2m} - 1 \right) > 1$. We now use Proposition \ref{P1} directly,
\begin{equation*}
\frac{D^2}{A} \geq \frac{1}{130}\left( \frac{L}{2m} - 1 \right) > \frac{1}{130} > \frac{1}{12^2} > \frac{D^2}{A},
\end{equation*}
which shows the contradiction and proves Theorem \ref{MAIN3}.

\vs

\noindent {\bf Acknowledgements}. The authors would like to thank Dmitry Korotkin and Hermann Nicolai for very helpful discussions. This work was part of J.P. PhD project, funded by CAP PhD scholarship and partially funded by Fondo Clemente Estable Project FCE\_1\_2023\_1\_175902 and by CSIC grant C013-347.

\providecommand{\noopsort}[1]{}\providecommand{\singleletter}[1]{#1}%


\begin{thebibliography}{10}

\bibitem{Myers:1987qx}
Robert~C. Myers.
\newblock {Superstring Gravity and Black Holes}.
\newblock {\em Nucl. Phys. B}, 289:701--716, 1987.

\bibitem{KN}
D.~Korotkin and H.~Nicolai.
\newblock A periodic analog of the schwarzschild solution.
\newblock {\em pre-print,\text{arXiv}}, 1994.

\bibitem{Wei92}
Gilbert Weinstein.
\newblock The stationary axisymmetric two-body problem in general relativity.
\newblock {\em Communications on Pure and Applied Mathematics},
  45(9):1183--1203, 1992.

\bibitem{vanStockum:1937zz}
Willem~Jacob van Stockum.
\newblock {The gravitational feild of a distribution of particles rotating
  about an axis of symmetry}.
\newblock {\em Proc. Roy. Soc. Edinburgh A}, 136:176–192, 1937.

\bibitem{Lewis}
Thomas Lewis.
\newblock {Some special solutions of the equations of axially symmetric
  gravitational fields}.
\newblock {\em Proc. R. Soc. Lond.}, 57:135--154, 1932.

\bibitem{Peraza_2023}
Javier Peraza, Martín Reiris, and Omar~E Ortiz.
\newblock Periodic analogues of the kerr solutions: a numerical study.
\newblock {\em Classical and Quantum Gravity}, 40(17):175010, aug 2023.

\bibitem{wald84}
R.M. Wald.
\newblock {\em General Relativity}.
\newblock University of Chicago Press, 1984.

\bibitem{AST_2008__321__195_0}
Piotr~T. Chrusciel and Jo\~ao~Lopes Costa.
\newblock On uniqueness of stationary vacuum black holes.
\newblock In Hijazi Oussama, editor, {\em G\'eom\'etrie diff\'erentielle,
  physique math\'ematique, math\'ematiques et soci\'et\'e (I) : Volume en
  l'honneur de Jean Pierre Bourguignon}, number 321 in Ast\'erisque, pages
  195--265. Soci\'et\'e math\'ematique de France, 2008.

\bibitem{Reiris1}
M.~Reiris.
\newblock On the existence of charged electrostatic black holes in arbitrary
  topology.
\newblock {\em pre-print,\text{arXiv}}, 2023.

\bibitem{Reiris:2015zaa}
Martin Reiris.
\newblock {On static solutions of the Einstein-Scalar Field equations}.
\newblock {\em Gen. Rel. Grav.}, 49(3):46, 2017.

\end{thebibliography}
\end{document}